\documentclass[a4paper, 11 pt]{article}
\usepackage[disableTODOs]{notes}
\usepackage{graphicx}
\usetikzlibrary{arrows}

\title{Matrix recovery from bilinear and quadratic measurements}
\author{Michalina Pacholska, Karen Adam, Adam Scholefield, Martin Vetterli}
\date{\today}

\def\keywords{\vspace{.5em}
{\textit{Keywords}:\,\relax%
}}

\begin{document}
\maketitle

\begin{abstract}
Matrix (or operator) recovery from linear measurements is a 
well-studied problem. However, there are situations where only bilinear or quadratic measurements are available. A bilinear or quadratic problem can easily be transformed into a linear one, but it raises questions when the linearized problem is solvable and what is the cost of linearization.

In this work, we study a few specific cases of this general problem and show when the~bilinear problem is solvable. Using this result and certain properties of polynomial rings, we present a scenario when the~quadratic problem can be linearized at the cost of just a linear number of additional measurements. 

Finally, we link our results back to two applications that inspired it:
Time Encoding Machines and Continuous Localisation.
\end{abstract}


\section{Introduction}
Given a sufficient number of linear equations, the problem of matrix recovery (or completion) does not differ from any other linear problem. However, when the problem is ill-posed, regularisations specific to matrix recovery are used. For example, a~low-rank assumption is often used in matrix completion~\cite{candes2009exact}, which leads to non-convex problems that are often solved with convex relaxation~\cite{candes2010power}.

Matrix recovery also appears in multidimensional signal estimation, for example in multi-channel audio processing or in source separation~\cite{Cardoso1996equivariant, Vincent2006performance}, where the signal coefficients form a matrix with different rows corresponding to different dimensions/sources. 

Such problems are rarely stated as matrix reconstruction from \emph{bilinear} 
measurements not because of their structure (which often fits the bilinear formulation), but because they can be seen as a 
special case of matrix recovery from linear measurements in the larger space. 
Consider for example  an unknown matrix \(\vect{C}: \mathbb{R}^{J\times K}\) that we would like to reconstruct. We can treat \(\vect{C}\) as a vector -- an element of \(\mathbb{R}^{JK} \), and try to solve a linear system of equations for \(\vect{C}\). 

There are situations when this approach is sufficient, for example if 
the measurements have the form \(b_{jk}=\vect{\phi}_j^\top \vect{C} 
\vect{\psi}_k\), where \(\vect{\phi}_j\) and \(\vect{\psi}_k\) are bases of \(\mathbb{R}^J\) and 
\(\mathbb{R}^K\) respectively. In this case, to reconstruct \(\vect{C}\) we need 
all \(JK\) possible measurements (without any regularizers or priors). However, having a certain number of measurements is not sufficient to guarantee recovery. Indeed, consider the situation when \(\vect{\psi}\) is a frame consisting 
of \(JK\) vectors. Then there are \(JK\) measurements of the form 
\(b_{k}=\vect{\phi}_0^\top \vect{C} \vect{\psi}_k\), \(k=0 \dots JK-1\), but 
they are insufficient to recover \(\vect{C}\), at least if there are no assumptions on the structure of~\(\vect{C}\).

A second problem we consider is matrix reconstruction from quadratic 
measurements, where \(J=K\) and the measurements have the~form 
\(b_{k}=\vect{\psi}_k^\top \vect{C} \vect{\psi}_k\), with \(\vect{\psi}_k \in 
\mathbb{R}^K\). Quadratic measurements appear for~example in phase retrieval \cite{miao1998phase, fienup1982phase}.

We introduce the problem and assumptions used through this paper in~Section~\ref{sec:problem}.
In Section \ref{sec:bilinear}, we focus on bilinear measurements. We analyze if the set of measurements is sufficient to reconstruct the matrix in the case when one set of vectors is 
a frame and the~second is a set of pair-wise different vectors. We apply 
our theory to~the problem of encoding mixed bandlimited signals, studied in the 
context of Time Encoding Machines (TEMs) \cite{adam2019encoding}.

In Section \ref{sec:quadratic}, we consider a combination of bilinear and 
quadratic measurements, with additional assumptions on the measurement vectors 
\(\vect{\psi}_k\). These assumptions are based on the properties of polynomials, or more generally, polynomial rings. 

Finally, we show how the quadratic case applies to continuous localization from range measurements \cite{pacholska2019relax}. We introduce these applications only briefly, and for details refer the reader to corresponding publications. We think, however, that applications help motivate the assumptions we make in this work, which may seem arbitrary without further context. 
\section{Problem statement}
\label{sec:problem}
In this section we introduce the most general problem we consider, show how to linearize it and introduce the specific assumptions we make in this work. The problem, as well as the assumptions, are inspired by the Continuous Localisation \cite{pacholska2019relax}.

We consider the problem of recovering a matrix \(\vect{C}\in \mathbb{R}^{J\times K}\) from \(N\) measurements of the form
\begin{equation}
\label{eq:problem_statement}
b_n = \vect{g}_n^\top \vect{C}\vect{f}_n + \vect{f}_n^\top \vect{L}\vect{f}_n,
\end{equation}
where \( \vect{g}_n\) and \(\vect{f}_n\) are known vectors in \(\mathbb R^J\) 
and \(\mathbb R^K\), respectively, \(b_n\) are measured scalars and \(\vect{L} \in \mathbb{R}^{K\times K}\) is an unknown matrix that does not need to be recovered.
In this work we only consider the noiseless case.

This problem can be also interpreted as a problem of recovering a bilinear operator \(C(\vect{f},\vect{g}) = \vect{g}^\top \vect{C} \vect{f}\) in the presence of the quadratic term \(L(\vect{f}) = \vect{f}^\top \vect{L} \vect{f}\). In this work, we only use the matrix representation.

We transform \eqref{eq:problem_statement} into a system of linear equations using properties of the trace. Since both elements of the sum in \eqref{eq:problem_statement} are scalars we can write \(\tr(\vect{g}_n^\top \vect{C}\vect{f}_n)=\tr((\vect{g}_n^\top \vect{C})^\top \vect{f}_n^\top)=\tr(\vect{C}^\top \vect{g}_n \vect{f}_n^\top)=\vectorised(\vect{g}_n \vect{f}_n^\top)^\top\vectorised(\vect{C})\) and similarly for the quadratic part. We then obtain a set of \(N\) linear equations:
\begin{equation}
    b_n = \vectorised(\vect{g}_n \vect{f}_n^\top)^\top\vectorised(\vect{C})
    +  \vectorised(\vect{f}_n \vect{f}_n^\top)^\top\vectorised(\vect{L}),
\end{equation}
where the linear transformation \(\vectorised : \mathbb{R}^{J\times K} \rightarrow \mathbb{R}^{JK}\) flattens a matrix into a~vector.

\subsection{General Assumptions}
\label{sec:assumptions}
Of course, measurements have to be pairwise 
different, that is there are no \(n\), \(m\) such that \(\vect{g}_n=\vect{g}_m\) 
and \(\vect{f}_n = \vect{f}_m\). 
We use stronger assumptions. In 
particular, we assume vectors \(\vect{f}_n\) can be parameterized by one variable \(t \in \mathbb R\), which we will call 
time. More precisely, we assume the \(k\)-th entry of \(\vect{f}_n\) has the 
form
\begin{equation}
\label{eq:model_form}
[\vect{f}_n]_k = f_k(t_n),
\end{equation}
where \(f_k : \mathcal{I} \rightarrow \mathbb R\)
\(k=0, \dots, K-1\) are \emph{linearly independent} functions from a linear space of 
functions~\(\mathcal F\), \(\mathcal I \in \mathbb R\) is an interval or the whole real line and  and \(t_n \in \mathcal I\), \(n=0\dots N-1\) are sampling times. 

Moreover, we assume that the sampling times \((t_0, 
\dots, t_{N-1})\) follow a continuous probability 
distribution on \(\mathcal{I}^N\) and that for every non-zero element \(f \in 
\mathcal F\), the set of zeros of \(f\) has Lebesgue measure (\(\lambda\)) 
equal to zero: \(\lambda(\{t| f(t)=0\})=0\).\footnote{It can be shown that if \(\mathcal{F}\) contains a constant function, then the above assumptions on \(\mathcal{F}\) and \(t_n\) guarantee that the vectors \(\vect{f}_n\) are pairwise different with probability one.}

Finally, let us note that we could consider \(\vect{C} \in \mathbb{C}^{J\times K}\), \(\vect{g}_n \in \mathbb{C}^J\) and \(\vect{f}_n \in \mathbb{C}^K\) and the results presented in this paper would still apply.


\section{Bilinear Measurements}
\label{sec:bilinear}

In this section we only consider the bilinear measurement problem -- we assume that there is no quadratic term (\(\vect{L}=0\)). Therefore we want to find \(\vect{C}\) such that 
\begin{equation}
b_n = \vect{g}_n^\top \vect{C}\vect{f}_n.
\end{equation}
Equivalently, we want to solve the following system of equations:
\begin{equation}
    b_n = \vectorised(\vect{g}_n \vect{f}_n^\top)^\top \vectorised(\vect{C}).
\end{equation}

This system of equations can be solved if at least \(JK\) of them are independent, in other words if there are \(JK\) independent vectors \(\vectorised(\vect{g}_n \vect{f}_n^\top)\). Theorem \ref{thm:bilinear} below states when it is the case, under the following additional assumptions on \(\vect{g}_n\).

Since we assume that the vectors \(\vect{f}_n\) are different, we allow vectors 
\(\vect{g}_n\) to repeat. Intuitively, we would like the vectors \(\vect{g}_n\) to be either linearly independent or equal. More formally, let 
\(\mathcal{A} \) be the set of unique vectors \(\vect{g}_n\), so \(\vect{g}_n \in \mathcal{A}\) for all \(n=0\dots N-1\) but \(M := |\mathcal{A}| \leq N\). Let \(\vect{a}_0, \dots, \vect{a}_{M-1}\) be the (unique) elements of \(\mathcal{A}\).
We will assume 
that every \(J\) elements of \(\mathcal{A}\) are linearly independent, or 
equivalently that every \(J\) elements of \(\mathcal{A}\) form a basis in 
\(\mathbb{R}^J\).

Under these assumptions, the following theorem holds:
\begin{restatable}[Basis of Bilinear Measurements]{thm}{mixed}
\label{thm:bilinear} 
Consider the set of \(N=KJ\) vectors of the form
\(\vectorised(\vect{g}_n\vect{f}_n^\top)\). It is a basis in~\(\mathbb{R}^{KJ}\) 
if and only if no more than \(K\) vectors \(\vect{g}_n\) are equal.
\end{restatable}

Before we prove Theorem \ref{thm:bilinear}, we have to introduce a few tools, including two lemmas, that we will prove in the last part of this section.

First, observe that in this theorem we assume that \(N=JK\). Clearly, any number less than this is not sufficient to reconstruct \(\vect{C}\). On the other hand, when \(N>JK\) we need only \(JK\) of the measurements to be independent,\footnote{In fact we will never get more than \(JK\) independent equations.} so we need only \(JK\) measurements to satisfy Theorem \ref{thm:bilinear}. 

We prove the case \(N=JK\), because it lets us use the properties of the determinant. Indeed, let 
\begin{equation}
\label{eq:basic_matrix}
\vect{\Gamma} = 
\begin{bmatrix}
\vectorised(\vect{g}_0\vect{f}_0^\top)^\top \\
\vectorised(\vect{g}_1\vect{f}_1^\top)^\top \\
\vdots \\
\vectorised(\vect{g}_{N-1}\vect{f}_{N-1}^\top)^\top
\end{bmatrix}.
\end{equation}
For \(N=JK\), \(\vect{\Gamma}\) is a square matrix, and thus its rows are independent if and only if its determinant is not zero. To calculate the determinant of \(\vect{\Gamma}\), we will use the specific structure of \(\vect{\Gamma}\), depicted in Figure \ref{fig:structure}.

Second, we use some properties of permutations. 
Recall that a permutation \(\vect{\sigma}\) of numbers \(0, \dots, N-1\) is 
any sequence (or vector, so we use bold letters to denote permutations) of length \(N\) in which each number from \(0\) to \(N-1\) (included) 
appears exactly once. We will call \(\mathcal{P}_N\) the set of all 
permutations of numbers  \(0, \dots, N-1\). By \(\sgn(\vect{\sigma})\) we will
denote the parity of the permutation \(\vect{\sigma}\).

Now assume that \(N=JK\). We will consider permutations \(\vect{\sigma}\) and \(\vect{\pi}\) to be equivalent if after grouping together the first \(J\) elements of the permutations, the next 
\(J\) elements of the permutation and so on, we get the same sets, see Figure 
\ref{fig:permutations}. More formally, we use the definition below:
\begin{definition}[\(\sim_{J}\)]
Two permutations \(\vect{\sigma}, \vect{\pi} \in 
\mathcal{P}_N\) are equivalent,  \(\vect{\sigma} \sim_J \vect{\pi}\), if and only if
\begin{align}
\forall_{k \in \{0, \dots, K-1\}} 
\{\sigma_j : \floor{j/J} = k\} =\{\tau_j : \floor{j/J} = k\}.
\end{align}
Note that \(\sim_{J}\) is symmetric, reflexive and transitive, so it is a proper equivalence relation.
\end{definition}

\begin{figure*}
\centering
\includegraphics[width=1\linewidth]{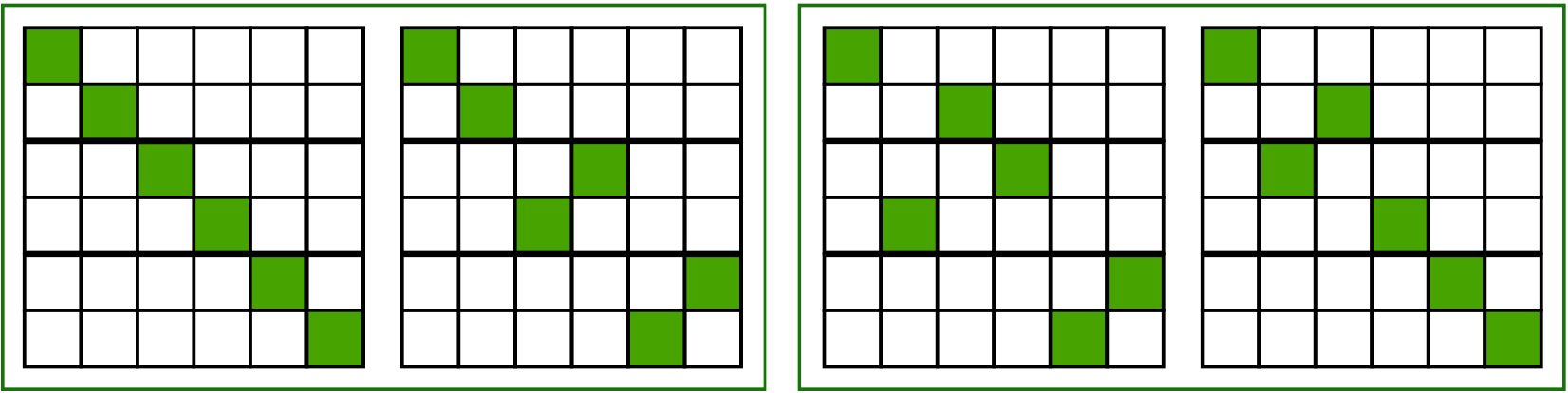}
\caption{Example of permutations of length 6. The rows correspond to the indexes \(n\) and the columns correspond to values of the permutation, so the dark squares represent pairs \((n, \sigma_n)\). 
The first permutation is the trivial/identity permutation (\(\sigma_n = n\)). The first two
permutations are equivalent with respect to \(\sim_2\) and the last two are 
equivalent with respect to \(\sim_2\), but all four are not equivalent.}
\label{fig:permutations}
\end{figure*}

We will also consider the equivalence class of \(\mathcal{P}_N\) by \(\sim_J\), 
i.e. the set of all permutations that are equivalent to the permutation
\(\vect{\sigma}\)
\begin{equation}
\left[\vect{\sigma} \right]:=\{\vect{\pi} \in \mathcal{P}_N | \vect{\pi} \sim_J 
\vect{\sigma} \}.
\end{equation}
The set of all equivalence classes is called a quotient set of \(\mathcal{P}_N\) 
by \(\sim_J\), and we denote it by \(\mathcal{P}_N / \sim_J\).

It is sometimes convenient to identify the equivalence class with one of its
elements, which is often called a \emph{representative}. This allows us to 
easily extend functions defined on permutations to classes of equivalence.
Formally, a representative can be defined using a \emph{selection function} \(\vect{s}:
\mathcal{P}_N / \sim_J \rightarrow \mathcal{P}_N\), such that 
\(\vect{s}([\vect{\sigma}]) \in [\vect{\sigma}]\). 

In our case, we choose the lexicographically first element of \([\vect{\sigma}]\) as its representative.
Using the representative, we can define \(\sgn({[\vect{\sigma}]})\) as 
\(\sgn({\vect{s}([\vect{\sigma}])})\), and \([\vect{\sigma}]_n\) as the \(n\)-th 
element of \(\vect{s}([\vect{\sigma}])\) etc.

All these properties of permutations are useful in the following lemma.
\begin{restatable}[Splitting determinant over permutations]{lemma}{permutations}
\label{lemma:permutations}
For a square matrix \(\vect{M}\in \mathbb R^{JK\times JK}\) and fixed 
permutation class \([\vect{\sigma}]\), we have
\begin{equation}
\label{eq:determinant_statement}
\det(\vect{M})
= \!\!\!\!\!\!\! \smash{\sum_{[\vect{\sigma}] \in \mathcal{P}_{JK}/\sim_J} 
\!\!\!\!\!\!\!  \sgn([\vect{\sigma}])} \prod_{k=0}^{K-1} 
\det([\vect{M}]_{[\vect{\sigma}], k}), 
\end{equation}
where by \([\vect{M}]_{[\vect{\sigma}],k}\) we denote a \(J\times J\) submatrix
of \(\vect{M}\) defined by rows \(Jk, Jk+1, \dots, J(k+1)-1\) and columns 
\([\vect{\sigma}]_{Jk}, [\vect{\sigma}]_{Jk+1}, \dots, [\vect{\sigma}]_{J(k+1)-1}\).
\end{restatable}

For the proof of Theorem \ref{thm:bilinear} we will also make use of an additional lemma stated below. This lemma is quite technical and we need it to control the measure of zeros of the determinant.

\begin{restatable}[Measure zero]{lemma}{measurezero}
\label{lemma:measurezero}
Consider a linear space of functions \(\mathcal{F}\) from \(\mathbb R\) to 
\(\mathbb R\), measurable with respect to the Lebesgue 
measure \(\lambda\) on \(\mathbb 
R\) such that for every non-zero element \(f \in \mathcal{F}\)
the set of zeros of \(f\) has measure zero. 
Now  let \(h : \mathbb R^K \rightarrow \mathbb R\) be a finite 
sum of products of non-zero element \(f_k \in \mathcal{F}\):
\begin{equation*}
h(t_0, \dots, t_{K-1}) 
= \sum_{\vect{\sigma} \in \mathcal P_K}\alpha_{\vect{\sigma}}
\prod_{k=0}^{K-1} f_{\sigma_k}(t_{k}),
\end{equation*}
where \(f_k\), \(k=0, \dots K-1\) form a linearly independent set in 
\(\mathcal F\) and \(\alpha_{\vect{\sigma}}\) are some constant coefficients.

Then, the set of zeros of \(h\) either has measure zero (with respect to the 
Lebesgue measure \(\lambda^{K}\)) or all the coefficients 
\(\alpha_{\vect{\sigma}}\) are zero.
\end{restatable}

We now prove Theorem \ref{thm:bilinear} assuming  Lemmas \ref{lemma:permutations} and \ref{lemma:measurezero}. We provide the proofs of the lemmas later, in Section \ref{sec:technical}.

\begin{proof}[Proof of Theorem \ref{thm:bilinear}]

We would like to prove a condition under which the rows of \(\vect{\Gamma}\) (defined in \eqref{eq:basic_matrix}) are 
independent, or 
equivalently when it is full row 
rank. First, observe that \(\vect{\Gamma}\) is a square matrix under the 
assumption that \(N=KJ\), and therefore it is full rank if and only if its 
determinant is not zero. 

By definition, the determinant of \(\vect{\Gamma}\) is
\begin{equation}
\label{eq:def_determinant}
\det(\vect{\Gamma}) = \sum_{\vect{\sigma}\in \mathcal{P}_N} \sgn(\vect{\sigma}) 
\prod_{n=0}^{N-1} \left[\vect{\Gamma}\right]_{\sigma_n, n}.
\end{equation}
As depicted in Figure \ref{fig:structure}, we can express one row of \(\vect{\Gamma}\) as
\begin{equation*}
\vectorised(\vect{g}_{n}\vect{f}_n^\top) = 
\vectorised \begin{bmatrix}
g_{1, n}f_0(t_n) & \dots & g_{1, n}f_{K-1}(t_n) \\
g_{2, n}f_0(t_n) & \dots & g_{2, n}f_{K-1}(t_n) \\ 
\vdots & \ddots & \vdots \\
g_{J-1, n}f_0(t_n) & \dots & g_{J-1, n}f_{K-1}(t_n) 
\end{bmatrix},
\end{equation*}
where we use 
\(f_i(t_n):=\left[\vect{f}_n\right]_i\) and \(g_{i, n}:= 
\left[\vect{g}_{n}\right]_i\) in order to simplify notation. 

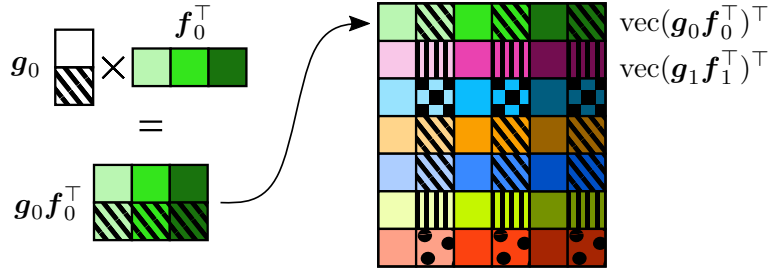
\begin{figure*}
\centering
        \begin{tikzpicture}
        \node[inner sep=0pt] (picture) at (0,0)
            {\includegraphics[width = 7.3cm]{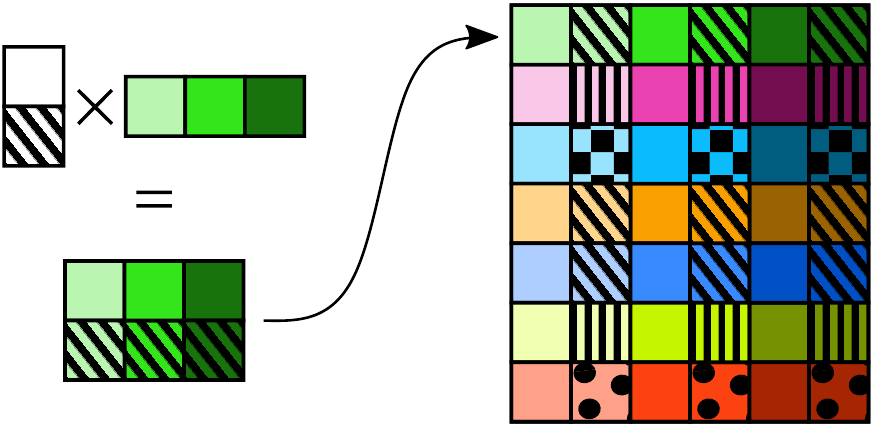}};
        \node (f) at (-1.8,1.5) {\(\vect{f}_0^\top\)};
        \node (g) at (-4,0.9) {\(\vect{g}_0\)};
        \node (fg) at (-3.7, -0.9) {\(\vect{g}_0\vect{f}_0^\top\)};
        \node (vect) at (4.8, 1.5) {\(\text{vec}(\vect{g}_0\vect{f}_0^\top)^\top\)};
        \node (vect) at (4.8, 0.9) {\(\text{vec}(\vect{g}_1\vect{f}_1^\top)^\top\)};
        \end{tikzpicture}
\caption{Structure of matrix \(\vect{\Gamma}\). For \(J=2\) and \(K=3\), the 
 \(\text{vec}(\vect{g}_n\vect{f}_n^\top)\) has length \(6\), and \(\vect{\Gamma}\) is a tall matrix for \(N>JK\). Observe, that in each row of \(\vect{\Gamma}\) every second element is a product of the second element of \(\vect{g}_n\) and some element of \(\vect{f}_n\), and similarly the first two elements of each row are products of the first element of \(\vect{f}_n\) and some element of \(\vect{g}_n\). Because of this structure, many products will repeat in the determinant.}
\label{fig:structure}
\end{figure*}

Note that the \(i\)-th element of \(\vectorised(\vect{g}_n\vect{f}_n)\) is a 
product of \(f_k(t_n)\) where \(k = \floor{i/J}\) and \(g_{j,n}\) where \(j=i 
\mod J\). Equivalently, 
by replacing row index \(i\) with multi-index \((j, k)\) such that \(i=kJ+j\), 
we can write
\begin{equation}
\label{eq:vectorised}
\left[\vectorised(\vect{g}_{n}\vect{f}_n^\top)^\top\right]_{kJ+j} = 
g_{j, n}f_k(t_n).
\end{equation}

From Lemma \ref{lemma:permutations}, we know that the determinant of 
\(\vect{\Gamma}\) can be described via determinants of its submatrices
\([\vect{\Gamma}]_{\vect{\sigma},k}\). For fixed \(k\) and \(\vect{\sigma}\), 
using the notation from \eqref{eq:vectorised}, we get
\begin{equation*}
[\vect{\Gamma}]_{\vect{\sigma},k} = \sum_{\vect{\tau} \in \mathcal{P}_J} 
\sgn(\tau) \prod_{j=0}^{J-1} g_{j, \sigma_{Jk+\tau_j}} 
f_k(t_{\sigma_{Jk+\tau_j}}),
\end{equation*}
where we replaced \(n\) with \(\sigma_{Jk+\tau_j}\).
Since \(k\) is fixed, we can factor out \(f_k(\cdot)\) and obtain
{\medmuskip=0mu
\thinmuskip=0mu
\thickmuskip=0mu
\begin{equation}
\label{eq:factorised}
[\vect{\Gamma}]_{\vect{\sigma},k} = \left(\sum_{\vect{\tau} \in \mathcal{P}_J} 
\sgn(\vect{\tau})
\prod_{j=0}^{J-1} g_{j, \sigma_{Jk+\tau_j}}\right)\left(\prod_{j=0}^{J-1} 
f_k(t_{\sigma_{Jk+j}}) \right),
\end{equation}
}
where we drop \(\vect{\tau}\) in the last brackets, because the expression 
\(\prod_{j=0}^{J-1} f_{k}(t_{\sigma_{kJ + j}})\) does not depend on the order
of the elements of the product.

Since it is true for every \(k\), then for \(\pi \in \mathcal{P}_{JK}\) we have:
\begin{equation*}
\prod_{k=0}^{K-1}\prod_{j=0}^{J-1} f_{k}(t_{{[\vect{\sigma}]}_{kJ + j}}) = \prod_{k=0}^{K-1}\prod_{j=0}^{J-1} f_{k}(t_{\pi_{kJ + j}})
\end{equation*}
if for each \(k\), \(\{[\vect{\sigma}]_{kJ + j}\}=\{\pi_{kJ + j}\} \), i.e. if \(\pi \sim_{J} \sigma\). On the other hand,
if  \(\vect{\pi}\) and \(\vect{\sigma}\) are not equivalent, the expressions might be equal at certain points, but are not equal \emph{everywhere}.
Thus, we get one expression for each class of equivalence of \( \sim_{J} \).

We can now group equal terms together and obtain the determinant in 
the following form:
\begin{equation}
\label{eq:determinant}
\det(\vect{\Gamma}(t_0, \dots, t_{N-1})) = 
\smashoperator{\sum_{[\sigma] \in \mathcal{P}_{JK}}} \gamma_{[\vect{\sigma}]}
\prod_{k=0}^{K-1}\prod_{j=0}^{J-1} f_{k}(t_{[\vect{\sigma}]_{kJ + j}}),
\end{equation}
where we explicitly state the dependence of \(\vect{\Gamma}\) on times \(t_n\) and where \(\gamma_{[\vect{\sigma}]}\) are constant coefficients, i.e. depending only 
on the vectors \(\vect{g}_n\) and permutation class \([\vect{\sigma}]\), but 
not on the parameter \(t\).

From \eqref{eq:factorised} and Lemma \ref{lemma:permutations}, we get that
\begin{equation*}
\gamma_{[\vect{\sigma}]} = 
\sgn([\vect{\sigma}])\prod_{k=0}^{K-1} \!\!\left(\smash{\sum_{\vect{\tau}\in 
\mathcal{P}_J}} \smash{\sgn(\vect{\tau})} \prod_{j=0}^{J-1} 
g_{j,[\vect{\sigma}]_{kJ + \tau_j}}\right)\!\!.
\end{equation*}
The expressions in brackets resembles the definition of the determinant.
Indeed, it is \(\det(\vect{G}_{\vect{[\sigma]}, k})\) where 
\(\vect{G}_{[\vect{\sigma}], k}\) is a matrix created by concatenating the vectors 
\(\vect{g}_{[\vect{\sigma}]_{kJ + j}}\):
\begin{equation*}
\vect{G}_{[\vect{\sigma}], k} = 
\begin{bmatrix} 
\vertbar & \vertbar &        & \vertbar \\
\vect{g}_{[\vect{\sigma}]_{kJ}} & \vect{g}_{[\vect{\sigma}]_{kJ + 1}} & \dots &
\vect{g}_{[\vect{\sigma}]_{(k+1)J-1}} \\
\vertbar & \vertbar &        & \vertbar
\end{bmatrix}.
\end{equation*}

Therefore,
\begin{equation*}
\gamma_{[\vect{\sigma}]} = 
\prod_{k=0}^{K-1} \det(\vect{G}_{[\vect{\sigma}],k}).
\end{equation*}

Now, each element of the sum in \eqref{eq:determinant} is a different function 
of \(\vect{t}\). From Lemma \ref{lemma:measurezero}, we know that either the set 
of zeros of \(\det({\vect{\Gamma}(\vect{t})})\) is measure zero, or all 
coefficients \(\gamma_{[\sigma]}\) are zero. 
The whole coefficient is zero if and only if at least one of the factors of this coefficient is zero:
\begin{equation}
\label{eq:quantifiers}
\forall_{[\vect{\sigma}] \in 
\mathcal{P}_N / \sim_J} \exists_{k\in \{0, \dots, K-1\}} \,
\det(\vect{G}_{[\vect{\sigma}], k}) = 0.
\end{equation}

What is left is to determine when \(\det(\vect{G}_{[\vect{\sigma}], k}) = 0\).
We have assumed that if all 
\(\vect{g}_{\sigma_{kJ + j}}\), \(j=0, \dots J-1\) are different they are independent, and therefore \(\det(\vect{G}_{[\vect{\sigma}], k})\) is not zero. This means that for 
\(\det(\vect{G}_{[\vect{\sigma}], k})\) to be zero, one of the vectors 
\(\vect{g}_m\) has to repeat. Equation \eqref{eq:quantifiers} is then equivalent 
to the the following statement: for each partition of \(JK\) measurements into 
\(J\)-element subsets, there exists a vector \(\vect{g}_n\) that appears at 
least twice in at least one of the measurement subsets. From the pigeonhole 
principle, for \eqref{eq:quantifiers} to hold, there has to be at least \(K+1\) equal vectors among vectors \({g}_n\). 
\end{proof}

\subsection{Technical proofs}
\label{sec:technical}
In this sections we prove Lemmas \ref{lemma:permutations} and \ref{lemma:measurezero}. Let us restate the lemmas before proving them:

\permutations*

\begin{proof}
From the definition, the determinant of an \(N\times N\) matrix is a certain sum over all permutations in \(\mathcal{P}_N\). Therefore, if \(N=JK\) it can be split into an external sum over classes of abstraction \(\mathcal{P}_{JK}/\sim_J\) and an internal sum over permutations in \([\vect{\sigma}]\):
\begin{equation*}
\det(\vect{M}) = \smashoperator{\sum_{\vect{\sigma}\in \mathcal{P}_N}} \sgn(\vect{\sigma})
\prod_{n=0}^{N-1} \left[\vect{M}\right]_{\sigma_n, n} 
= \sum_{[\vect{\sigma}]\in \mathcal{P}_{JK}/\sim_J}
{\sum_{\vect{\pi}\in [\vect{\sigma}]}} \sgn(\vect{\pi})
\prod_{n=0}^{N-1} \left[\vect{M}\right]_{{\pi_n}, n}
\end{equation*}
From now on, we will consider only the internal sum. From the definition of \(\sim_J\), if \(\pi \in [\vect{\sigma}]\) we have
\begin{equation*}
    \{[\vect{\sigma}]_{Jk},\dots, [\vect{\sigma}]_{Jk+(J-1)} \} = \{\pi_{Jk},\dots, \pi_{Jk+(J-1)}\},
\end{equation*}
for any \(k=0\dots K-1\). Therefore, we can permute each block of \(J\) elements of \([\vect{\sigma}]\) separately to obtain \(\vect{\pi}\). More formally, we can identify \(\vect{\pi} \in [\vect{\sigma}]\)  with \(K\) permutations \(\vect{\tau}^{0}, \dots \vect{\tau}^{K-1} \in \mathcal{P}_J\), such that
\begin{equation*}
    [\vect{\sigma}]_{Jk + \tau^k_j} = \pi_{Jk + j}, 
\end{equation*}
where we use the upper index to distinguish the \(k\)-th permutation from the \(k\)-th element of a permutation. 

Then, by identifying \(k=\floor{n/J}\) and \(j= n \mod J\)  we can write the internal sum over \(\vect{\pi} \in [\vect{\sigma}]\) as
\begin{gather*}
\sum_{\vect{\pi}\in [\vect{\sigma}]} \sgn(\vect{\pi})
\prod_{k=0}^{K-1}\prod_{j=0}^{J-1} [\vect{M}]_{kJ + j,\pi_{kJ + j}} \\
= \sgn([\vect{\sigma}])
\!\!\!\sum_{\vect{\tau}^0\in \mathcal{P}_J}\!\!\!\!\sgn(\vect{\tau}^0)
\dots \!\!\!\!\!
\!\!\!\sum_{\vect{\tau}^{K-1}\in \mathcal{P}_J}\!\!\!\!\!\sgn(\vect{\tau}^{K-1})
\prod_{k=0}^{K-1}\prod_{j=0}^{J-1} [\vect{M}]_{kJ + j,[\vect{\sigma}]_{kJ + \tau_j^k}}.
 \end{gather*}
Observe that each of the \(K\) products over \(j=0\dots J-1\) depends only on one 
\(\vect{\tau}^k\). Thus, for example the first product can be factored out before all sums except the sum over \(\vect{\tau}^0\):
{
\medmuskip=0mu
\thinmuskip=0mu
\thickmuskip=0mu
\small{
\begin{equation*}
\sgn([\vect{\sigma}])
\smash{\sum_{\vect{\tau}^0\in \mathcal{P}_J}}\mkern-6mu\sgn(\vect{\tau}^0) \prod_{j=0}^{J-1} 
[\vect{M}]_{j,[\vect{\sigma}]_{\tau_j^0}}
\dots \mkern-12mu \sum_{\vect{\tau}^{K-1}\in \mathcal{P}_J} \mkern-12mu \sgn(\vect{\tau}^{K-1})
\prod_{k=1}^{K-1}\prod_{j=0}^{J-1} [\vect{M}]_{kJ + j,[\vect{\sigma}]_{kJ + \tau_j^k}}.
\end{equation*}}}
Repeating this operation \(K\) times (for each \(\vect{\tau}_k\)), we get
\begin{equation*}
\sgn([\vect{\sigma}])\prod_{k=0}^{K-1} \left(\sum_{\vect{\tau}^k\in \mathcal{P}_J}
\sgn(\vect{\tau}^k) \prod_{j=0}^{J-1} 
[\vect{M}]_{kJ + j, [\vect{\sigma}]_{kJ + \tau_j^k}}\right),
\end{equation*}
where the expression in brackets is exactly \(\det(\vect{M}_{[\vect{\sigma}],j})\) as defined in the statement of the lemma.
\end{proof}

\measurezero*

\begin{proof}
We will prove Lemma \ref{lemma:measurezero} by induction over \(K\). Assume that Lemma \ref{lemma:measurezero} is true for \(K-1\), and assume that not all coefficients \(\alpha_{\vect{\sigma}}\) are zero.

Observe that \(h\) can be split in to \(K\) elements:
\begin{align*}
h(t_0, \dots, t_{K-1}) 
&= \sum_{\vect{\sigma} \in \mathcal P_K}\alpha_{\vect{\sigma}}
\prod_{k=0}^{K-1} f_{\sigma_k}(t_{k}),\\
&= \sum_{j=0}^{K-1} f_j(t_{K-1}) \left(\sum_{\substack{\vect{\sigma}\in \mathcal{P}_K \\ \sigma_{K-1}=j}}{\alpha_{\vect{\sigma}}} \prod_{k=0}^{K-2}f_{\sigma_k}(t_k)\right),
\end{align*}
where each of the \(K\) expressions in brackets has the same form as \(h\), but for \(K\) one smaller.
Indeed, by re-indexing functions \(f_k\) with indexes up to \(K-2\) (differently for different \(j\)) we can remove the condition \(\sigma_{K-1}=j\) and obtain a sum over permutations in \(\mathcal{P}_{K-1}\). 

If not all coefficients \(\alpha_{\vect{\sigma}}\) are zero, then there is \(\hat \jmath\) such that not all coefficients \(\alpha_{\vect{\sigma}}\) for \(\sigma_{K-1}= \hat\jmath\) are zero. From the Lemma for \(K-1\), the measure of the set of zeros (\(\mathcal{Z}_{\hat \jmath} \in \mathbb{R}^{K-1}\)) of the expression in brackets number \(\hat \jmath\) is zero (\(\lambda^{K-1}(\mathcal{Z}_{\hat \jmath})=0\)). The set of points for which all the expressions in brackets are zero (\(\mathcal{Z}_{all}\)) is a subset of \(\mathcal{Z}_{\hat \jmath}\), so it also has measure zero (\(\lambda^{K-1}(\mathcal{Z}_{all})=0\)), and therefore, from the definition of the Lebesgue measure, \(\lambda^{K}(\mathcal{Z}_{all}\times \mathbb R) = 0\) as well. 

For a fixed point \(\vect{s} \in \mathcal R^{K-1}\), \(\vect{s} \not \in \mathcal{Z}_{all}\), the function \(h(\vect{s}, t_{K-1})\) is a non-constantly-zero function of variable \(t_{K-1}\) that belongs to \(\mathcal{F}\), therefore the set of zeros of \(h(\vect{s},t_{K-1})\), \(\mathcal Z_{\vect{s}}\) is of measure zero.
By Fubini's Theorem, we calculate the total measure of the set of points \((\vect{s},t_{K-1})\) such that  \(h(\vect{s},t_{K-1}) = 0 \) by integrating the measure of \(\mathcal Z_{\vect{s}}\) over \(\vect{s} \not \in \mathcal{Z}_{all}\):
\begin{equation*}
    \smashoperator{\int_{\mathbb{R}^{K-1} / \mathcal{Z}_{all}}} \lambda(Z_{\vect{s}})d\vect{s} =  \smashoperator{\int_{\mathbb{R}^{K-1} / \mathcal{Z}_{all}}} 0 d\vect{s} = 0.
\end{equation*}
Thus, the total measure of zeros of \(h\) is zero for \(K\), and we have proven the induction step.

To complete the proof, we take base case \(K=0\), where \(h=\alpha_{0}\) is a constant. Then, either \(\alpha_0\) is zero and all coefficients are zero, or \(\alpha_0\) is not zero, the set of zeros of \(h\) is empty, and it has measure zero.
\end{proof}
\subsection{Application: Mixed Time Encoding}
We now apply Theorem~\ref{thm:bilinear} to provide reconstruction guarantees in the setting of mixed time encoding. For more information regarding this mixed time encoding setup, see~\cite{adam2019encoding}.

Consider a continuous-time vector signal $\vect{x}(t)$ with $N$ components $x^{(j)}(t), j = 1 \cdots J$ that are each bandlimited to $\left[ -\Omega, \Omega\right]$, such that each $x^{(j)}(t)$ can be written as a sum of $K$ $\sinc$ functions: 
\begin{equation}\label{eq: mTEM x^j definition}
    x^{(j)}(t) = \sum_{k=1}^K c_{jk}\sinc_\Omega(t - t_k),
\end{equation}
where $\sinc_\Omega(t) = \sin(\Omega t)/(\pi t)$ and $t_k = k\pi/\Omega + t_0$ for some fixed and known $t_0$. It follows that the signal is in $L^2(\mathbb{R})$ and has a well-defined integral $X^{(j)} (t):= \int_{-\infty}^t x^{(j)}(u)\, du < \infty $.

Assume $\vect{x}(t)$ is sampled as follows. First, it is passed through a mixing matrix $\vect{A} \in \mathbb{R}^{I\times J}$ producing the output $\vect{y}(t) = \vect{Ax}(t)$. Then, each of the signals $y^{(i)}(t), i=1\cdots I$ is sampled using a time encoding machine with parameters $\kappa^{(i)}$, $\delta^{(i)}$ and $b^{(i)}$, starting at time $t_0^{(i)}$, with known initial values of their integrators $\zeta_0^{(i)} = -\kappa^{(i)}\delta^{(i)}$ . The time encoding machine $i$ records times $t_\ell^{(i)}$ that satisfy $$\int_{t_\ell^{(i)}}^{t_{\ell+1}^{(i)}} y^{(i)}(u) + b^{(i)} \, du = 2\kappa^{(i)}\delta^{(i)}.$$

Defining $Y^{(i)}(t) := \int_{t_0^{(i)}}^{t} y^{(i)}(u)\, du$, we notice that the obtained time samples essentially provide amplitude samples of $Y^{(i)}(t)$:
\begin{equation*}
    Y^{(i)} (t^{(i)}_\ell) = \sum_{m=1}^\ell  \int_{t^{(i)}_{m-1}}^{t^{(i)}_{m}} y^{(i)}(u) \, du.
\end{equation*}

On the other hand, starting from the definition of the $x^{(j)}(t)$'s in~\eqref{eq: mTEM x^j definition}, we can write $Y^{(i)} (t^{(i)}_\ell)$ as follows:
\begin{align}
     Y^{(i)} (t^{(i)}_\ell) &= \sum_{j=1}^J a_{ij} \sum_{m=1}^K c_{jk} \times \int_{-\infty}^{t^{(i)}_\ell}\sinc_\Omega(u - t_k)\, du \notag\\
     Y^{(i)} (t^{(i)}_\ell) &= \vect{a}_i^T \vect{C} \vect{f}^{(i)}_{\ell}, \notag
\end{align}
where $\vect{f}^{(i)}_\ell = \left[\int_{t^{(i)}_0}^{t^{(i)}_\ell}\sinc_\Omega(u - t_k )\, du\right]_{k}$, $a_{ij}$ are the elements of the mixing matrix $\vect{A}$ and $\vect{a}_i^T = [a_{i0}, a_{i1},\cdots , a_{iJ}]$.
We adopt a change of notation, let the measurements be indexed by $n$, so that $b_n:=Y^{(i)} (t^{(i)}_\ell)$ for a known couple $(i,\ell)$ which is unique for each $n$, $\vect{g}_n := \vect{a}_i$, and $\vect{f}_n := \vect{f}^{(i)}_{\ell}$. We obtain the following equation:
\begin{equation}\label{eq: mTEM as bilinear sampling}
    b_n = \vect{g}_n^T \vect{C} \vect{f}_n.
\end{equation}

Following this reformulation, we can use Theorem~\ref{thm:bilinear} to obtain the following corollary:
\begin{corollary}
Let $\vect{x}(t)$ be composed of $J$ components such that $x^{(j)}(t)$ is as defined in~\eqref{eq: mTEM x^j definition} for some fixed $\Omega$ and for $c_{jk}$'s sampled from a Lipschitz-continuous probability distribution. Then, sampling the $\vect{y}(t)=\vect{Ax}(t)$, where $\vect{A}\in\mathbb{R}^{I\times J}$, using $I$ time encoding machines that emit $n_{\textrm{spikes}}^{(i)}$ spikes, $i = 1\cdots I$, the original signals $\vect{x}(t)$ can be perfectly recovered from the signals with probability one if
\begin{equation}\label{eq: mTEM reconstructibility condition}
    \sum_{i=1}^{I} \min\left(n_{\textrm{spikes}}^{(i)}, K\right) \geq JK.
\end{equation}
\end{corollary}

\begin{proof}
The connection to Theorem~\ref{thm:bilinear} is made explicit in~\eqref{eq: mTEM as bilinear sampling}. It remains to be shown that~\eqref{eq: mTEM reconstructibility condition} is a sufficient condition for the assumptions of Theorem~\ref{thm:bilinear} to hold.

First, let us examine the condition on the $\vect{g}_n$'s. We note that~\eqref{eq: mTEM reconstructibility condition} ensures that at most $K$ of the $\vect{g}_n$ vectors are equal. Indeed, the number of spikes emitted by time encoding machine $i$ determines how many times the vector $\vect{g}_n$ takes the value $\vect{a_i}$ which is set by the mixing matrix $\vect{A}$. Since each $\vect{a_i}$ can be used at most $K$ times, we obtain the $\min$ term in~\eqref{eq: mTEM reconstructibility condition}. The summation over the different components $i$ then computes the total number of available samples when taking this constraint into account.

Second, let us examine the condition on the functions $f$. These functions are integrals of the sinc functions and are thus bandlimited. Therefore, they span a space of functions $\mathcal{F}$ which comprises only of bandlimited functions and the set of zeros of $f$ thus has Lebesgue measure $\lambda$ equal to zero as required: $\lambda \lbrace t| f(t) = 0 \rbrace = 0$.

Finally, let us examine the condition on the sampling times $t_\ell^{(i)}$. These are required to  follow a continuous probability distribution.

We start by noting that the sampling times are random because the signals are generated by sampling the $c_{jk}$'s from a Lipschitz continuous probability distribution.\footnote{This is a reasonable assumption which is satisfied for example by Gaussian and uniform distributions} For simplicity we will write $y^{(i)}(t) = \sum_{k=1}^K d_{ik} \sinc_\Omega(u-t_k)$, where $d_{ik} = \sum_{j=1}^{J} a_{ij} c_{jk}$. If we assume the $a_{ij}$'s are set, the coefficients $d_{ik}$ also follow a Lipschitz continuous probability distribution.

The cumulative density function for the sampling time $t_k$ given the time $t_{k-1}$ can therefore be written as
\begin{align}
    P&\left(t^{(i)}_\ell \leq t |t^{(i)}_{\ell-1}\right) = p\left( \int_{t^{(i)}_{\ell-1}}^{t} y^{(i)}(u)\, du \geq 2\kappa^{(i)}\delta^{(i)} - b^{(i)}(t-t^{(i)}_{\ell-1})\right) \label{eq: mTEM cdf}\\
    &= p\left( \sum_\ell d_{ik} \int_{t^{(i)}_{\ell-1}}^t\sinc_\Omega(u-t_k)\, du \geq 2\kappa^{(i)}\delta^{(i)} - b^{(i)}(t-t^{(i)}_{\ell-1})\right)
\end{align}

Now consider the random variables $d_{ik}$, they follow a probability distribution $p\left(d_{ik}\leq \gamma\right)$ which is Lipschitz continuous with respect to $\gamma$. Now if we define $\gamma:=h_{ik}(t)$ where $h_{ik}(t)$ is absolutely continuous with respect to $t$, then $p\left(d_{ik}\leq h_{ik}(t)\right)$ will also be absolutely continuous with respect to $t$~\cite{royden2010real}. 

Now let us define $h_{ik}(t)= 2\kappa^{(i)}\delta^{(i)} b(t-t^{(i)}_{\ell-1})/\left( \int_{t^{(i)}_{\ell-1}}^t \sinc_\Omega(u-t_k)\, du\right)$. It is then absolutely continuous and we can write
\begin{align}
    p\left(d_{ik}\leq h_{ik}(t)\right) & = p\left(d_{ik} \leq 2\kappa^{(i)}\delta^{(i)} b(t-t^{(i)}_{\ell-1})/\left( \int_{t^{(i)}_{\ell-1}}^t \sinc_\Omega(u-t_k)\, du\right)\right)\notag\\
    & = p\left(d_{ik}\left( \int_{t^{(i)}_{\ell-1}}^t \sinc_\Omega(u-t_k)\, du\right) \leq 2\kappa^{(i)}\delta^{(i)} b(t-t^{(i)}_{\ell-1}) \right)\notag\\
    &=: p_{ik}(t) \notag
\end{align}

We now notice that the cumulative distribution function in~\eqref{eq: mTEM cdf} can be expressed as a convolution of the functions $(1-p_{ik}(t))$, all of which are absolutely continuous. Therefore, the cumulative distribution function on the $t^{(i)}_\ell$'s is absolutely continuous.
\end{proof}


\section{Quadratic measurements}
\label{sec:quadratic}

In this section, we reintroduce the quadratic term from \eqref{eq:problem_statement}. We introduce additional assumptions that let us analyse the quadratic term separately (Observation\,\ref{obs:polynomials}) and show how both terms can be connected (Lemma~\ref{lemma:polynomial_stacking}). Finally we show how to use the aforementioned results in practice, using the example of trajectory reconstruction from \cite{pacholska2019relax} (Section \ref{sec:trajectory}).

Recall that the entries \(\vect{f}_n\) are elements of the linear space of functions \(\mathcal F\). In this section, we will additionally assume that \(\mathcal F\) can be extended to a polynomial ring \(\mathcal R\) over \(\mathbb R\), \(\mathcal F \subset \mathcal R\).

This assumption might seem abstract, but it encompasses a number of widely used linear spaces of functions.
A canonical example is of course the family of polynomials \(f_k(t) = 
t^k\). We can set \(\mathcal F\) to be the space of polynomials of degree 
smaller than \(K\), and \(\mathcal R\) to be ring of polynomials 
\(\mathcal{R}[t]\). Trigonometric polynomials, that is  real symmetric
bandlimited functions on \((-\pi, \pi)\), can be extended to \(\mathcal R= \mathcal{R}[\cos(t)]\), and similarly real bandlimited functions on 
\((-\pi, \pi)\) can be extended to  \(\mathcal R = \mathcal{R}[X, Y]/_{[X^2+Y^2-1]}\), where  we identify \(X=\cos(t)\) and \(Y=\sin(t)\).
We can also take \(\mathcal F\) to be the space of complex bandlimited functions on \((-\pi, \pi)\) and \(\mathcal R\) to be \(\mathcal{R}[e^{\iu t}]\)

Let us now consider the purely quadratic term of \eqref{eq:problem_statement}:
\begin{equation}
\tilde{b}_n = \vect{f}_n^\top \vect{L}\vect{f}_n.
\end{equation}
Since the matrices \(\vect{f}_n\vect{f}_n^\top\) are symmetric, it is clear that we can have only \(K(K+1)/2\) independent equations, and can recover \(\vect{L}\) only up to a subspace (without any additional assumptions on \(\vect{L}\)). It turns out that for polynomial rings the maximal number of independent measurements is even smaller, as per the observation below:

\begin{restatable}[On Polynomial Rings]{observation}{polynomials}
\label{obs:polynomials}
Assume that \(\deg(f_k) \leq \vect{\alpha}\), for some \(\vect{\alpha} \in \mathbb N^m\).
Then the entries of \(\vect{f}_n \vect{f}_n^\top\) are also polynomials over 
\(\mathbb{F}\), of degrees \(\leq 2\vect{\alpha}\), because \(\deg(ab) = 
\deg(a) + \deg(b)\) for \(a, b \in \mathcal R\) or \(\deg(ab) = \deg(a) + \deg(b)\)
up to the appropriate modulo relation for a quotient ring.

This means that, among \(M\) different vectors \(\vectorised(\vect{f}_n 
\vect{f}_n^\top)\) of degrees up to \(\vect{\alpha}\), there are \textbf{at 
most} \((a_0 +1)(a_1+1)(a_2+1)\dots\) linearly independent vectors.

If additionally for every degree \(\vect{\beta} < \vect{\alpha}\) there is an 
index \(k\) such that \(\deg(f_k)=\vect{\beta}\), then each degree \(\leq 
2\vect{\alpha}\) will have a corresponding entry in \(\vect{f}_n \vect{f}_n^\top\).
This means that for \(f_k(t) = t^k\) there are \textbf{exactly}  \(\min(M, 
2K-1)\) linearly independent vectors. 
For a polynomial of two variables, with no equivalence relation and with 
\(\vect{\alpha} = (\sqrt{K}, \sqrt{K})\), we would get \(4K\) different degrees 
(so a \(4\)-fold increase).

For trigonometric polynomials, if \(f_k\) have degrees up to \(\vect{\alpha}=(1, 
(K-1)/2)\) (where \(K\) is odd), then the maximal degree we can get would be 
\(2\vect{\alpha}=(2, K-1)\), which would lead to \(3K\) different degrees.
But since for trigonometric polynomials we have the relation \(X^2 + Y^2 = 1\), 
the highest possible degree we can have is \((1, K-1)\), which reduces the 
possible number of different degrees to \(2K\). For a standard basis of the 
space of real, periodic bandlimited functions with bandwidth \(K\), we do not 
use polynomials of degree exactly \((1, (K-1)/2)\) (the function \(\sin(t)\cos(t(K-1)/2)\) 
is not the element of the standard basis).
Thus we never get the degree \((1, K-1)\) and the maximal number of different degrees is
\(2K-1\).
\end{restatable}

This means that we rarely have enough measurements to recover (even symmetric) \(\vect{L}\). However, the low number of degrees of freedom can be advantageous if we consider both terms from \eqref{eq:problem_statement}, and are primarily interested in recovering \(\vect{C}\). We will see a specific example of this in Section \ref{sec:trajectory}.

Recall the linear system of equations corresponding to \eqref{eq:problem_statement}:
\begin{equation}
    b_n = \vectorised(\vect{g}_n \vect{f}_n^\top)\vectorised(\vect{C})
    +  \vectorised(\vect{f}_n \vect{f}_n^\top)\vectorised(\vect{L}).
\end{equation}
Theorem \ref{thm:bilinear} gives us a condition on the first term of the equation, and in this section we have analysed what is maximum number of degrees of freedom of the second term. The following lemma shows how we can combine the results.

\begin{restatable}[Expanding a matrix with polynomials]{lemma}{polynomialstacking}
\label{lemma:polynomial_stacking}
Let \(\vect{A}_i \in \mathbb R^{r\times r}\) be a full rank matrix and
let \(\vect{A}_{i+1}\) be a matrix constructed by first appending any column to
\(\vect{A}_i\) and then appending a row of the form
\begin{equation*}
\begin{bmatrix}
p_0(t) & \dots & p_{r-1}(t) &  p_{r}(t) 
\end{bmatrix},
\end{equation*}
where \(p_j\in \mathcal{R}\) is evaluated at a random time \(t\) (from a
continuous distribution) and the degree of \(p_{r}\) in at least one of the
variables is greater than the degree of any other \(p_j\) in the same variable.
Then, \(\vect{A}_{i+1}\) is full rank with probability one.
\end{restatable}

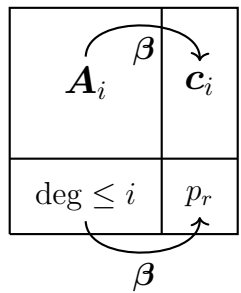
\begin{figure}[ht]
\centering
\begin{tikzpicture}[thick]
\draw (0,0) -- (3,0) -- (3,3) -- (0,3) -- (0,0);
\draw (0,1) -- (3,1);
\draw (2,0) -- (2,3);
\node[] (1) at (1, 2) {\Large{\(\vect{A}_i\)}};
\node[] (2) at (1, 0.5) {\large{\(\deg \leq i\)}};
\node[] (3) at (2.5, 0.5) {\large{\(p_r\)}};
\node[] (4) at (2.5, 2) {\Large{\(\vect{c}_i\)}};
\path[->]
(1) edge[bend left=90] node [below] {\large{\(\vect{\beta}\)}} (4)
(2) edge[bend right=90] node [below] {\large{\(\vect{\beta}\)}} (3);
\end{tikzpicture}
\caption{Building matrix \(\vect{A}_{i+1}\) from matrix \(\vect{A}_i\). For \(\vect{A}_{i+1}\) to not be full rank, \(p_r(t)\) has to be related to the rest of the row via the same combination \(\vect{\beta}\) as \(\vect{c}_i\) relates to \(\vect{A}_i\)}
\label{fig:ai}
\end{figure}

\begin{proof}
Let \(\vect{c}_i\) be the appended column.
Since \(\vect{A}_{i}\) is full rank, there is a unique linear combination 
\(\vect{\beta}\) such that \(\vect{A}_{i}\vect{\beta} = \vect{c}_i\)
see Figure \ref{fig:ai}. For \(\vect{A}_{i+1}\) to not be full rank, would mean 
that the same linear combination \(\vect{\beta}\) of the added row would have to 
be equal to the last diagonal element of \(\vect{A}_{i+1}\), \(p_r(t)\). We 
could write it as
\begin{equation}
\label{eq:row}
p_r(t) - 
\begin{bmatrix}
p_0(t) & \dots & p_{r-1}(t)
\end{bmatrix}\vect{\beta} = 0.
\end{equation}
The left hand side of this equation is not constantly zero, because there is a
variable in which \(\deg(p_r)\) is bigger than \(\deg(p_i)\), 
\(i=0, \dots r-1\) in this variable. Since we assumed that times follow a 
continuous distribution, and the measure of the set of zeros of any non-zero polynomial is 
zero, we get that the probability that \eqref{eq:row} is satisfied is zero. 
Therefore, \(\vect{A}_{i+1}\) is full rank with probability one.
\end{proof}
\subsection{Application: Continuous Localisation}
\label{sec:trajectory}
In this section we see how to use Lemma \ref{lemma:polynomial_stacking} in practice.
We consider continuous localisation from distance measurements as our example.
The problem from \cite{pacholska2019relax} is defined as follows.

\begin{figure}[ht]
    \centering
    \begin{tikzpicture}
        \node[above] (picture) at (0,0){
        \includegraphics[width=5cm]{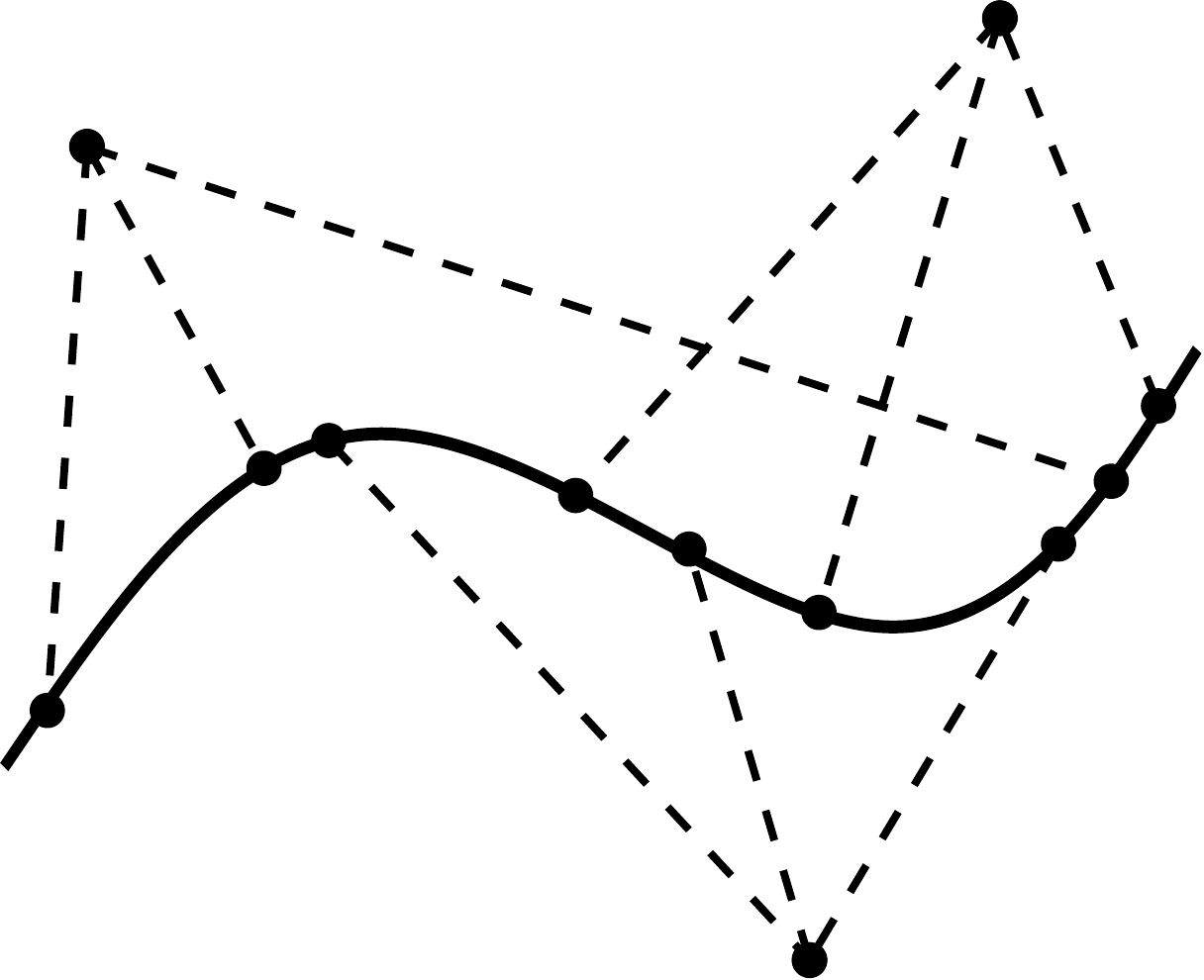}};
        \put(-60,110){$\vect{a}_0$}
        \put(30,0){$\vect{a}_1$}
        \put(-53,40){$\vect{r}(t)$}
        \put(55,120){$\vect{a}_2$}
    \end{tikzpicture}
    \caption{Continuous localisation: the device moves on trajectory \(\vect{r}\) and measures distances to anchor points \(\vect{a}_0, \vect{a}_1\) and \(\vect{a}_2\), at non uniform times. In such a setup, simple lateration will not work, because at no point are the distances to all three anchors known. Figure from \cite{pacholska2019relax}.}
    \label{fig:trajectory}
\end{figure}
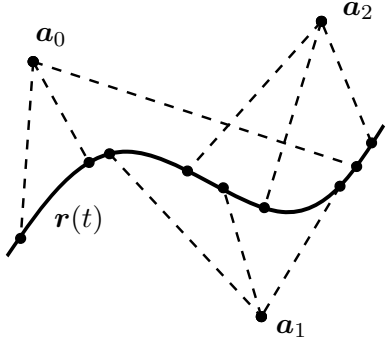

We consider a device moving on a trajectory \(\vect{r} : \mathcal{I} \rightarrow \mathbb{R}^D\)
and at each time \(t_n\), \(0 \dots N-1\) distance measurements are taken to \emph{one of} the known anchors \(\vect{a}_m\), \(m=0 \dots M-1\), see Figure \ref{fig:trajectory}.
\begin{equation}
\label{eq:trajectory}
    d_n =  \|\vect{a}_{m_n} - \vect{r}(t_n)\|.
\end{equation}
For standard lateration to work, the positions of the anchors \(\vect{a}_m\) cannot lie in the same affine subspace. Our assumptions are slightly stronger -- that no \(D+1\) anchors lie in the same affine subspace.

Moreover, we assume that the robot trajectory coordinates belong to a \(K\)-dimensional 
linear space of functions \(\mathcal{F}\), and thus we can write:
\begin{equation}
\label{eq:trajectory_model}
    \vect{r}(t_n) = \vect{C}\vect{f_n}
\end{equation}
where \(\vect{f_n}\) are vectors defined as in \eqref{eq:model_form}. We assume that \(\mathcal{F}\) is either the space of polynomials of degree up to \(K\) or the space of periodic bandlimited functions. Both these spaces satisfy the assumptions from Section \ref{sec:assumptions} and can be extended to certain polynomial rings, see Section \ref{sec:quadratic}.
To recover the trajectory we need to recover a \(D\times K\) matrix of its coefficients \(\vect{C}\). For details, see \cite{pacholska2019relax}.

In the noiseless case, we can write this problem as a system of quadratic equations of the form \eqref{eq:problem_statement}:
\begin{equation}
\label{eq:trajectory_linear}
b_n  = \vect{a}_{m_n}^\top {\vect{C}}\vect{f}_n -  \frac{1}{2}
\vect{f}_n^\top\vect{L}\vect{f}_n.
\end{equation}
where \( {\vect{L}} =  {\vect{C}}^\top  {\vect{C}}\), and with  \(b_{n}:= \frac{1}{2}\left(||\vect{a}_{m_n}||^2 - d^2_n\right) \).

We then drop the relation between \(\vect{C}\) and \(\vect{L}\) and obtain a linear system of equations. Note, that any solution to \eqref{eq:trajectory} with \eqref{eq:trajectory_model} is also a solution to \eqref{eq:trajectory_linear}. Thus the \emph{unique} solution to the linearization \eqref{eq:trajectory_linear} solves the original problem.

From Observation \ref{obs:polynomials}, Theorem \ref{thm:bilinear} and Lemma \ref{lemma:polynomial_stacking} we get the following Corollary.
\begin{corollary}[Theorem 1 from \cite{pacholska2019relax}]
\label{corollary:trajectory} 
Given \(N \geq K(D+2) - 1\) measurements (at different 
times), the matrix \(\vect{C}\) can be uniquely recovered with probability one 
if: 
\begin{equation}
\sum_{m=0}^{M-1} \min(k_m, K) \geq K(D+1), 
\label{eq:full_condition} 
\end{equation}
where \(k_m\) is the number of 
measurements in which the \(m\)-th anchor is used. Moreover, if Condition 
\eqref{eq:full_condition} is not satisfied, \(\vect{C}\) cannot be uniquely 
reconstructed by solving the linear system of equations.
\end{corollary}

We now provide a quick sketch of the proof of Corollary \ref{corollary:trajectory}. For the full proof see \cite{pacholska2019relax}.
First, from Observation \ref{obs:polynomials}, we know that there can be at most \(2K-1\) linearly independent vectors \(\vect{f}_n \vect{f}_n^\top\). Thus, if the vectors \(\vect{a}_{m_n} \vect{f}_n^\top\) from the bilinear part of \eqref{eq:trajectory_linear} are independent, then on the whole we can have at most \(K(D+2)-1\) independent equations, and any further measurements are redundant.

Second, identify \(\vect{g}_n = \begin{bmatrix}\vect{a}_{m_n} \\ 1\end{bmatrix}\) and \(J=D+1\). Thanks to the assumption  that no \(D+1\) anchors lie in the same affine subspace, the collection \(\vect{g}_n\) satisfies  the assumptions from Section \ref{sec:bilinear}. From Theorem \ref{thm:bilinear} we know that if among those \(N\) measurements there are \(K(D+1)\) measurements such that no anchor \(\vect{a}_{m}\) is used more than \(K\) times, then we can reconstruct \(\vect{C}\). If we calculate the number of such measurements available, we get exactly \eqref{eq:full_condition}.

Finally, we need to know that we have exactly \(K(D+2)-1\) linearly independent measurements, or equivalently that the blilinear and quadratic parts are linearly independent. This can be shown by applying Lemma \ref{lemma:polynomial_stacking} inductively over the elements of the quadratic part with \(\vect{A}_0 \in \mathbb{R}^{K(D+1)\times K(D+1)}\). Note, that since \(\vect{g}_n \) is defined by appending 1 to \(\vect{a}_{m_n}\), the matrix \(\vect{A}_0\) already contains the columns of the quadratic part of degree up to \(K-1\), so the added columns have degrees greater than \(\vect{A}_0\) and Lemma \ref{lemma:polynomial_stacking} applies.

\section*{Contributions}
AS and MP formulated the initial problem; MP generalised the problem, formulated and proved the results and wrote the manuscript; KA and MP found the connection to TEMs, KA formalised this connection and wrote Section 3.2. MV and AS advised the research.

\bibliography{main} 

\begin{thebibliography}{1}

\bibitem{candes2009exact}
E.~J. Cand{\`e}s and B.~Recht, ``Exact matrix completion via convex
  optimization,'' {\em Foundations of Computational mathematics}, vol.~9,
  no.~6, p.~717, 2009.

\bibitem{candes2010power}
E.~J. {Candes} and T.~{Tao}, ``The power of convex relaxation: Near-optimal
  matrix completion,'' {\em IEEE Transactions on Information Theory}, vol.~56,
  pp.~2053--2080, May 2010.

\bibitem{Cardoso1996equivariant}
J.~. {Cardoso} and B.~H. {Laheld}, ``Equivariant adaptive source separation,''
  {\em IEEE Transactions on Signal Processing}, vol.~44, pp.~3017--3030, Dec
  1996.

\bibitem{Vincent2006performance}
E.~{Vincent}, R.~{Gribonval}, and C.~{Fevotte}, ``Performance measurement in
  blind audio source separation,'' {\em IEEE Transactions on Audio, Speech, and
  Language Processing}, vol.~14, pp.~1462--1469, July 2006.

\bibitem{miao1998phase}
J.~Miao, D.~Sayre, and H.~Chapman, ``Phase retrieval from the magnitude of the
  fourier transforms of nonperiodic objects,'' {\em JOSA A}, vol.~15, no.~6,
  pp.~1662--1669, 1998.

\bibitem{fienup1982phase}
J.~R. Fienup, ``Phase retrieval algorithms: a comparison,'' {\em Applied
  optics}, vol.~21, no.~15, pp.~2758--2769, 1982.

\bibitem{adam2019encoding}
K.~Adam, A.~Scholefield, and M.~Vetterli, ``Encoding and decoding mixed
  bandlimited signals using spiking integrate-and-fire neurons,'' {\em To
  appear in IEEE International Conference on Acoustics, Speech and Signal
  Processing (ICASSP)}, 2020.

\bibitem{pacholska2019relax}
M.~{Pacholska}, F.~{Duembgen}, and A.~{Scholefield}, ``Relax and recover:
  Guaranteed range-only continuous localization,'' {\em IEEE Robotics and
  Automation Letters}, February 2020.

\bibitem{royden2010real}
H.~L. Royden {\em et~al.}, {\em Real analysis}.
\newblock Prentice Hall,, 2010.

\end{thebibliography}
\bibliographystyle{ieeetr}

\end{document}